\documentclass[a4paper,twocolumn]{article}
\usepackage[utf8x]{inputenc}

\usepackage{amsmath}
\usepackage{comment}
\usepackage{amsfonts}

\usepackage{amsthm}
\newtheorem{theorem}{Theorem}{\bfseries}{\normalfont}
{\bfseries}{\normalfont}
\newtheorem{proposition}{Proposition}{\bfseries}{\normalfont}
{\bfseries}{\normalfont}
{\bfseries}{\normalfont}
\newtheorem{obs}{Observation}{\bfseries}{\normalfont}
\newtheorem{lemma}{Lemma}{\bfseries}{\normalfont}

\newcommand{\oldqed}{}

\usepackage{tikz}
\usetikzlibrary{arrows,decorations.pathmorphing,decorations.pathreplacing,backgrounds,positioning,fit,shapes,arrows,calc}

\usepackage{bbding}
\usepackage{placeins}

\usepackage{microtype,ellipsis}
\usepackage{enumerate}

\usepackage[pdfdisplaydoctitle,pagecolor=orange!40!black,menucolor=orange!40!black,filecolor=magenta!40!black,urlcolor=blue!40!black,linkcolor=red!40!black,citecolor=green!40!black,colorlinks]{hyperref} 

\hypersetup{pdftitle={On the Parameterized and Approximation Hardness of Metric Dimension}, pdfauthor={Sepp Hartung and André Nichterlein}}

\newcommand{\decprob}[3]{%
  \begin{center}%
    \begin{minipage}{0.9\linewidth}%
      \textsc{#1}\\
      \textbf{Input:} #2\\
      \textbf{Question:} #3
    \end{minipage}%
  \end{center}%
}

\usepackage{units}

\usepackage[numbers,sort]{natbib}
\makeatletter
\def\NAT@spacechar{~}
\makeatother

\usepackage{cleveref}
\usepackage{placeins}

\title{On the Parameterized and Approximation Hardness of \\ Metric Dimension}

\author{
Sepp Hartung and André Nichterlein \medskip \\  
  Institut f\"ur Softwaretechnik und Theoretische Informatik, \\ TU Berlin, Berlin, Germany \\
\small{\texttt{\{sepp.hartung,andre.nichterlein\}@tu-berlin.de}}}
\date{}

\usepackage{etoolbox}

\newcounter{Bew1}
\newcounter{Bew2}

\newcommand{\appendixproof}[2]{
   \label{L\arabic{Bew1}}
   \gappto{\appendixProofText}{ \phantomsection \stepcounter{Bew2}\label{\arabic{Bew2}} \subsection{Proof~\textbf{\arabic{Bew2}} (#1)} #2}
} 

\newcommand{\appendixproofT}[2]{
   \label{L\arabic{Bew1}}
   \gappto{\appendixProofTextT}{ \phantomsection \stepcounter{Bew2}\label{\arabic{Bew2}} \subsection{Proof~\textbf{\arabic{Bew2}} (#1)} #2}
}

\DeclareMathOperator{\dist}{dist}

\newcommand{\MD}{\textsc{Metric Dimension}\xspace}
\newcommand{\DomSet}{\textsc{Dominating Set}\xspace}
\newcommand{\BipDomSet}{\textsc{Bipartite Dominating Set}\xspace}
\newcommand{\RedBlueDomSet}{\textsc{Red-Blue Dominating Set}\xspace}

\renewcommand{\P}{\mathcal{P}}

\newcommand{\TL}{\operatorname{TL}}
\newcommand{\TR}{\operatorname{TR}}
\newcommand{\BL}{\operatorname{BL}}
\newcommand{\BR}{\operatorname{BR}}
\newcommand{\M}{\operatorname{M}}

\crefname{rrule}{Rule}{Rules}
\crefname{obs}{Observation}{Observations}

\newcommand{\C}{\mathcal{C}}

\newcommand{\size}{k}

\newcommand{\addFactor}{\ensuremath{\frac{3}{2}}}
\newcommand{\halfaddFactor}{\ensuremath{\frac{1}{4}}}
\newcommand{\addFactorMinusOne}{\ensuremath{\frac{1}{2}}}
\newcommand{\yValue}{\ensuremath{10n^2}}
\graphicspath{{images/}}

\begin{document}


\twocolumn[
  \begin{@twocolumnfalse}
    \maketitle
\begin{abstract}
\looseness=-1 The NP-hard \MD problem is to decide for a given graph~$G$ and a positive integer~$k$ whether there is a vertex subset of size at most~$k$ that \emph{separates} all vertex pairs in~$G$. 
Herein, a vertex~$v$ separates a pair $\{u,w\}$ if the distance (length of a shortest path) between $v$ and~$u$ is different from the distance of~$v$ and $w$. 
We give a polynomial-time computable reduction from the \BipDomSet problem to \MD on maximum degree three graphs such that there is a one-to-one correspondence between the solution sets of both problems. 
There are two main consequences of this: First, it proves that \MD on maximum degree three graphs is W[2\text{]}-complete with respect to the parameter~$k$. This answers an open question concerning the parameterized complexity of \MD posed by D\'{\i}az et~al.~[ESA'12\text{]} and already mentioned by Lokshtanov~[Dagstuhl seminar, 2009\text{]}.
Additionally, it implies that \MD cannot be solved in $n^{o(k)}$ time, unless the assumption $\text{FPT}\neq \text{W[1]}$ fails. This proves that a trivial $n^{O(k)}$ algorithm is probably asymptotically optimal.

Second, as \BipDomSet is inapproximable within $o(\log n)$, it follows that \MD on maximum degree three graphs is also inapproximable by a factor of $o(\log n)$, unless $\text{NP}=\text{P}$. This strengthens the result of Hauptmann et~al.~[JDA'12\text{]} who proved APX-hardness on bounded-degree graphs.
\bigskip

\end{abstract}
  \end{@twocolumnfalse}
]

\section{Introduction}
Given an undirected graph~$G=(V,E)$ a \emph{metric basis} of~$G$ is a vertex subset~$L\subseteq V$ such that each pair of vertices ~$\{u,w\}\subseteq V$ is \emph{separated} by~$L$, meaning that there is at least one~$v\in L$ such that $\dist(v,u)\neq \dist(v,w)$. 
Herein, ``$\dist(v,u)$'' denotes the length of a shortest path between~$v$ and~$u$. 
The corresponding \MD problem has been independently introduced by \citet{HM76} and \citet{Sla75}:%
\decprob{\MD \cite[GT61]{GJ79}}%
{An undirected graph~$G = (V,E)$ and an integer~$\size \geq 1$.}  
{Is there a metric basis of size at most~$k$?
}
The metric dimension of graphs (the size of a minimum-cardinality metric basis) finds applications in various areas including network discovery~\& verification~\cite{BEEHHMR06}, metric geometry~\cite{HM76}, robot navigation, coin weighing problems, connected joins in graphs, and strategies for the Mastermind game. 
We refer to  \citet{CHMPPSW07}, \citet{HMPSW10}, and \citet{BC11} for a more comprehensive list and a more complete bibliography on the extensive study on metric dimension.

There is a rich literature about the metric dimension of graphs, but little is known about the computational complexity of \MD.
It is known to be NP-hard and there is a linear-time algorithm for trees~\cite{KRR96}. 
From a polynomial-time approximation point of view, it has been shown to admit a $2\log n$-approximation~\cite{KRR96} and that it cannot be approximated within $o(\ln n)$, unless P=NP~\cite{BEEHHMR06}. \citet{HSV12}  showed that, unless $\text{NP}\subseteq \text{DTIME}(n^{\log\log n})$, there is no $(1-\epsilon)\ln n$ approximation algorithm for any $\epsilon>0$. Furthermore, they proved APX-hardness on bounded-degree graphs. \citet{DPSL12} showed that \MD remains NP-hard on bipartite graphs but becomes polynomial-time solvable on outerplanar graphs. Additionally, \citet{ELW12} examined the complexity of the vertex-weighted variant of \MD on various graph classes.

\paragraph{\bf Our Contribution.} 
We provide a polynomial-time computable reduction that maps an instance of \BipDomSet, consisting of a bipartite graph and an integer~$h$, to an equivalent \MD instance $(G,k)$ with $k=h+4$ and~$G$ having maximum degree three.

\looseness=-1 Since \BipDomSet is W[2]-hard~\cite{RS08Algorithmica}, our reduction proves that \MD is W[2]-hard with respect to~$k$ even on graphs with maximum degree three. Additionally, we prove W[2]-completeness. The question on the parameterized complexity of \MD (on general graphs) was posed by \citet{Lok09};
also \citet{DPSL12} pointed to this question. The W[2]-hardness of \MD showes that, unless the widely believed conjecture $FPT\neq W[2]$ fails, \MD is not \emph{fixed-parameter tractable}, that is, it cannot be solved within $f(k)\cdot |G|^{O(1)}$ time for any computable function~$f$. 
On the other hand, an algorithm that tests each size-$k$ subset being a metric basis runs in~$O(n^{k+2})$ time. 
However, our reduction together with the result that \BipDomSet cannot be solved in $n^{o(k)}$ time~\cite{CCF+05,RS08Algorithmica}, implies that \MD on an $n$-vertex graph cannot be solved in $n^{o(k)}$ time, unless FPT$=$W[1]. Thus the trivial $n^{O(k)}$-algorithm is (probably) asymptotically optimal.

\looseness=-1 Furthermore, as \DomSet cannot be approximated within a factor of~$o(\log n)$ unless $\text{NP}=\text{P}$~\cite{AS03}, it also follows that there cannot be an $o(\log n)$-factor approximation for \MD. This strengthens the APX-hardness result for bounded-degree graphs~\cite{HSV12} and it shows that the $2\log n$-approximation on general graphs~\cite{KRR96} is up to constant factors also optimal on bounded-degree graphs.

\paragraph{\bf Preliminaries.}
A problem that is shown to be \emph{W[1]- or W[2]-hard} is not fixed-parameter tractable, unless W[1] or W[2] is equal to the class FPT which consists of all fixed-parameter tractable problems. 
One can prove W[1]- or W[2]-hardness by means of a \emph{parameterized reduction} from a W[1]- or W[2]-hard problem. This is a mapping of an instance~$(I,k)$ of a problem~$\mathcal{A}$ in $h(k)\cdot |I|^{O(1)}$ time (for any computable~$h$) into an instance $(I',k')$ for~$\mathcal{B}$ such that $(I,k)\in \mathcal{A}\Leftrightarrow (I',k')\in \mathcal{B}$ and $k'\le g(k)$ for some~$g$. Our reduction is indeed a \emph{polynomial-time} computable parameterized reduction.
We refer to the monographs of \citet{DF99}, \citet{FG06}, and \citet{Nie06} for a detailed introduction to parameterized complexity analysis.

We use standard graph-theoretic notations. All the graphs are undirected and unweighted without self-loops. For a graph~$G=(V,E)$ with vertex set~$V$ and edge set~$E$ we set $n:=|V|$. A path~$P$ in~$G$ is a sequence of vertices $v_1-v_2-\ldots-v_s$ such that for $1\le i<s$ all $\{v_i,v_{i+1}\}\in E$. 
If there is a unique shortest path between two vertices~$u$ and~$v$, then we write just $u-v$ for this path without listing the intermediate vertices.
We write $\dist(v,u)$ for the length of a shortest path between $v$ and~$u$. Moreover, we set $\dist(v_1,v_2,\ldots,v_i)=\dist(v_1,v_2)+\dist(v_2,v_3)+\ldots+\dist(v_{i-1},v_i)$.

\paragraph{\bf Organization.} 
In the next section we describe our reduction and prove its correctness in~\autoref{sec:correctness}. We proceed by proving W[2]-completeness (\autoref{sec:w2-completness}) and, finally, in \autoref{sec:run-lower-bound} we prove the running time as well as the approximation lower bound. Some proofs are defered to the appendix.

\section{Construction of the Reduction}\label{sec:construction}
In this section we give a reduction from the W[2]-complete \BipDomSet problem~\cite{RS08Algorithmica} to \MD.
%
\decprob{\BipDomSet}%
{A bipartite graph~$G = (V_1\cup V_2,E)$ and an integer~$h \geq 1$.}  
{Is there a \emph{dominating set} of size at most~$h$, that is, a vertex subset~$V' \subseteq V_1\cup V_2$ such that $N[v]\cap V'\neq \emptyset$  for all $v\in V_1\cup V_2$?}
Since the distances from a vertex to all the vertices in a closed neighborhood of another vertex differ by at most three, $\log_3 \Delta$ is a lower bound on the metric dimension, where~$\Delta$ is the maximum degree. 
Avoiding large degrees was the main obstacle in the reduction below.

\begin{figure*}[ht]
	\begin{center}
		\def\Length{6}
		\def\LengthVertical{2}
		\def\Ydist{2.5}
		\def\yStretch{0.6}
		\begin{tikzpicture}[draw=black!75, scale=1,>=stealth']
			\tikzstyle{vertex}=[circle,draw=black!80,minimum size=6pt,inner sep=0pt]
			
			\node[vertex] (l1) at (-1,0.5) {$u^t_\ell$};
			\node[vertex] (l2) at (-1,-0.5) {};

			\node[vertex] (t-0) at (0,0) {};
			\foreach \i in {1,...,\Length}{
				\node[vertex] (t-\i) at (\i,0) {$u^t_{\i}$};
			}
			
			\pgfmathtruncatemacro{\value}{1 + \Length};	\let\A=\value
			\pgfmathtruncatemacro{\value}{2 + \Length};	\let\B=\value
			\pgfmathtruncatemacro{\value}{3 + \Length};	\let\C=\value
			\node[] (t-\A) at (\A,0) {$\ldots$};
			\node[vertex] (t-\B) at (\B,0) {$u^t_{n}$};
			\node[vertex] (t-\C) at (\C,0) {};

			\node[vertex] (r1) at (\C + 1,0.5) {$u^t_r$};
			\node[vertex] (r2) at (\C + 1,-0.5) {};

			\foreach \i in {1,...,\A}{
				\pgfmathtruncatemacro{\value}{1 + \i};	\let\nr=\value
				\path (t-\i) edge[very thick] (t-\nr);
			}
			\path (t-1) edge (t-0);
			\path (t-\B) edge (t-\C);
			\path (l1) edge (t-0);
			\path (l2) edge (t-0);
			\path (r1) edge (t-\C);
			\path (r2) edge (t-\C);

			\pgfmathtruncatemacro{\value}{5 + \Length};	\let\Xdist=\value
		\end{tikzpicture}
	\end{center}
\caption{Schematic illustration of the top-line of the skeletal structure. Bold edges indicate $y$-paths.}\label{fig:skeleton-line}
\end{figure*}
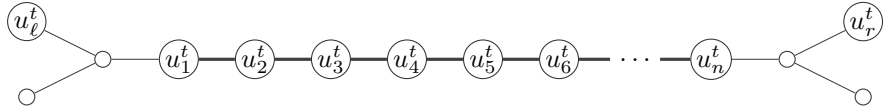

Let~$(G = (V_1\cup V_2,E), h)$ be a \BipDomSet-instance and let~$n := |V_1\cup V_2|$. We set $V:=V_1\cup V_2$ and fix a numbering $V_1=\{v_1,v_2,\ldots,v_s\}$ and $V_2=\{v_{s+1},\ldots,v_n\}$ such that for all $\{v_i,v_j\}\in E$ it holds that $j\ge i+3$. 
(Clearly, increasing $h$ by two, introducing two isolated vertices to~$V_1$, and numbering them by $v_{s+1}$ and~$v_{s+2}$ ensures this.)
We construct an equivalent instance~$(G' = (V',E'), k)$ of \MD with $k:=h+4$.

We start with a high-level description of the graph~$G'$: It consists of a \emph{skeletal structure} in which we insert a vertex-gadget for each vertex of~$G$ and an edge-gadget for each edge of~$G$.
Furthermore, there are four additional vertices in the skeletal structure that are forced to be in any metric basis and these four vertices separate all but $n$ vertex pairs in each vertex-gadget. Then, choosing a vertex in a vertex-gadget separates all of the $n$~vertex pairs in its own gadget plus all the pairs in the vertex gadgets that are ``adjacent'' by an edge-gadget to the chosen vertex-gadget.

Throughout the construction, several times we will connect two vertices~$\{u,v\}$ by a so-called \emph{$y$-path}, meaning that we insert a path of length~$y$ from~$u$ to~$v$. In all cases we make sure that the $y$-path is the unique shortest path between the endpoints and thus $u-v$ denotes this $y$-path. We set $y:=\yValue$ as we will assume that~$\halfaddFactor y>2n+2$. Intuitively, $y$-paths can be viewed as edges of weight~$y$.

We now describe the construction of~$G'$ in detail:
First, the skeletal structure is formed by $2n$ vertices $u^t_1,\ldots,u^t_n$ and $u^b_1,\ldots,u^b_n$ where all consecutive vertex pairs $\{u^t_{i},u^t_{i+1}\}$ and $\{u^b_i,u^b_{i+1}\}$ are connected by a $y$-path. 
We call the vertices~$u^t_1,u^t_n,u^b_1,u^b_n$ \emph{endpoints}. For each endpoint add a length-three path, a so-called $P_3$, and make the endpoint adjacent to the middle vertex. 
We call the first path $u^t_1-u^t_2-\ldots-u^t_n$ the \emph{top-line} and the second path $u^b_1-u^b_2-\ldots-u^b_n$ the \emph{bottom-line} both including the~$P_3's$. 
Additionally, let~$u^t_\ell$ be any degree-one vertex in the~$P_3$ attached to~$u^t_1$ and correspondingly let~$u^t_r,u^b_\ell,u^b_r$ be degree-one vertices in the~$P_3$'s attached  to~$u^t_n,u^b_1$, and $u^b_n$, respectively  (see \autoref{fig:skeleton-line}). (These are the four vertices that separate all but $n$ pairs in each vertex gadget.)

For each vertex~$v_i \in V$ we add the \emph{vertex-gadget}~$g^V_i$ to~$G'$ (see \autoref{fig:vertex-gadget}): 
Construct a cycle of length $2n+2$ and call two vertices~$a_i^t,a_i^b$ on the cycle with distance exactly $n+1$ the \emph{anchors} of the vertex gadget. Connect the \emph{top-anchor}~$a_i^t$ by a $y$-path to~$u^t_i$ and, symmetrically, connect the \emph{bottom-anchor} $a_i^b$ by a $y$-path to~$u^b_i$. 
There are two paths, each consisting of $n$~vertices between the anchors and we denote the vertices on these paths by $l^i_1,\ldots,l^i_n$ and $r^i_1,\ldots,r^i_n$, respectively.
The \emph{left-vertices} $l^i_1,\ldots,l^i_n$ remain degree-two vertices in~$G'$ whereas the \emph{right-vertices} $r^i_1,\ldots,r^i_n$ will be used in the following to connect the edge-gadgets.

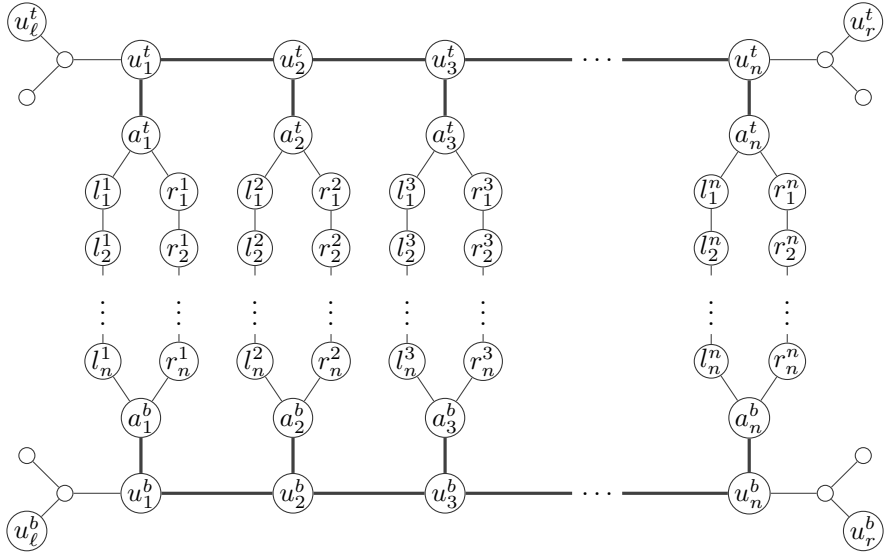
\begin{figure*}[t]
 	\begin{center}
		\def\Length{3}
		\def\LengthVertical{2}
		\def\Ydist{0}
		\def\yStretch{0.75}
		\def\xStretch{2}
		\begin{tikzpicture}[draw=black!75, scale=1,>=stealth']
			\tikzstyle{vertex}=[circle,draw=black!80,minimum size=6pt,inner sep=0pt]
			
			\pgfmathtruncatemacro{\value}{1 + \Length};					\let\A=\value
			\pgfmathtruncatemacro{\value}{2 + \Length};					\let\B=\value
			\pgfmathtruncatemacro{\value}{1 + \LengthVertical};			\let\AA=\value
			\pgfmathtruncatemacro{\value}{2 + \LengthVertical};			\let\BB=\value
			\pgfmathsetmacro{\value}{-2 -\yStretch - \yStretch*\BB};	\let\YPathDist=\value

			\foreach \i in {0,1} {
				\node[vertex] (l1\i) at (\xStretch - 1.5,0.5 + \i * \YPathDist) {\ifthenelse{\i=0}{$u^t_\ell$}{}};
				\node[vertex] (l2\i) at (\xStretch - 1.5,-0.5+ \i * \YPathDist) {\ifthenelse{\i=0}{}{$u^b_\ell$}};

				\node[vertex] (t-0-\i) at (\xStretch - 1,\i * \YPathDist) {};
				\foreach \j in {1,...,\Length}{
					\node[vertex] (t-\j-\i) at (\xStretch * \j,\i * \YPathDist) {\ifthenelse{\i=0}{$u^t_{\j}$}{$u^b_{\j}$}};
				}
				
				\pgfmathtruncatemacro{\value}{1 + \Length};	\let\A=\value
				\pgfmathtruncatemacro{\value}{2 + \Length};	\let\B=\value
				\pgfmathtruncatemacro{\value}{3 + \Length};	\let\C=\value
				\node[] 		(t-\A-\i) at (\xStretch * \A,\i * \YPathDist) {$\ldots$};
				\node[vertex] 	(t-\B-\i) at (\xStretch * \B,\i * \YPathDist) {\ifthenelse{\i=0}{$u^t_{n}$}{$u^b_{n}$}};
				\node[vertex] 	(t-\C-\i) at (\xStretch * \B + 1,\i * \YPathDist) {};
				
				\node[vertex] (r1\i) at (\xStretch * \B + 1.5,0.5 + \i * \YPathDist) {\ifthenelse{\i=0}{$u^t_r$}{}};
				\node[vertex] (r2\i) at (\xStretch * \B + 1.5,-0.5+ \i * \YPathDist) {\ifthenelse{\i=0}{}{$u^b_r$}};

				\foreach \j in {1,...,\A}{
					\pgfmathtruncatemacro{\value}{1 + \j};	\let\nr=\value
					\path (t-\j-\i) edge[very thick] (t-\nr-\i);
				}
				\path (t-1-\i) edge (t-0-\i);
				\path (t-\B-\i) edge (t-\C-\i);
				\path (l1\i) edge (t-0-\i);
				\path (l2\i) edge (t-0-\i);
				\path (r1\i) edge (t-\C-\i);
				\path (r2\i) edge (t-\C-\i);
			}
			
			\foreach \k in {1,...,\B} {
				\ifthenelse{\k=\A}{}
				{
					\pgfmathtruncatemacro{\value}{\xStretch * \k};	\let\Xdist=\value
					
					\node[vertex] (vt2\k) at (\Xdist,-1 + \Ydist) {\ifthenelse{\k=\B}{$a^t_n$}{$a^t_\k$}};

					\pgfmathtruncatemacro{\value}{1 + \LengthVertical};	\let\AA=\value
					\pgfmathtruncatemacro{\value}{2 + \LengthVertical};	\let\BB=\value

					\foreach \i in {1,2}{
						\foreach \j in {1,...,\LengthVertical}{
							\node[vertex] (v-\i-\j) at (\i + \Xdist - 1.5,-\yStretch*\j + \Ydist - 1)
								{%
									\ifthenelse{\i=1}{%
										{\ifthenelse{\k=\B}{$l^n_\j$}{$l^\k_\j$}}
									}{%
										{\ifthenelse{\k=\B}{$r^n_\j$}{$r^\k_\j$}}
									}%
								};
							
						}
						\node[] (v-\i-\AA) at (\i + \Xdist - 1.5,-\yStretch*\AA + \Ydist - 1) {$\vdots$};
						\node[vertex] (v-\i-\BB) at (\i + \Xdist - 1.5,-\yStretch*\BB + \Ydist - 1)
							{%
								\ifthenelse{\i=1}{%
									{\ifthenelse{\k=\B}{$l^n_n$}{$l^\k_n$}}
								}{%
									{\ifthenelse{\k=\B}{$r^n_n$}{$r^\k_n$}}
								}%
							};
					}

					\node[vertex] (vb1\k) at (\Xdist,-1 -\yStretch + \Ydist - \yStretch*\BB) {\ifthenelse{\k=\B}{$a^b_n$}{$a^b_\k$}};
					
					\foreach \i in {1,2}{
						\foreach \j in {1,...,\AA}{
							\pgfmathtruncatemacro{\value}{1 + \j};	\let\nr=\value
							\path (v-\i-\j) edge (v-\i-\nr);
						}
						\path (v-\i-1) edge (vt2\k);
						\path (v-\i-\BB) edge (vb1\k);
					}
					\path (t-\k-0) edge[very thick] (vt2\k);
					\path (vb1\k) edge[very thick] (t-\k-1);
				}
			}
		\end{tikzpicture}
	\end{center}
\caption{Schematic illustration of vertex-gadgets and their embedding into the skeletal structure. Bold edges indicate $y$-paths.}
\label{fig:vertex-gadget}
\end{figure*}


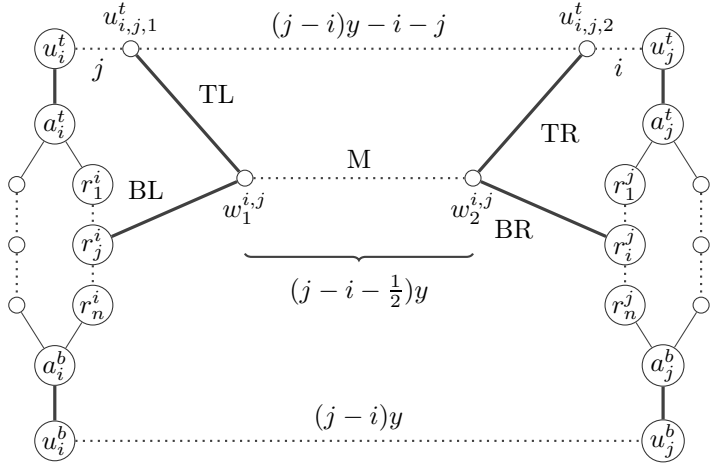
\begin{figure*}[t]
	\begin{center}
		\def\Xdist{8}
		\def\Ydist{0}
		\def\yStretch{0.8}
		\def\distWToGadget{2.5}
		\begin{tikzpicture}[auto,draw=black!75, scale=1,>=stealth',decoration=brace]
			\tikzstyle{vertex}=[circle,draw=black!80,minimum size=6pt,inner sep=0pt]
			\tikzstyle{vertex2}=[circle,draw=black!80,minimum size=15pt,inner sep=0pt]

			\pgfmathsetmacro{\value}{-2 -\yStretch * 4};	\let\YPathDist=\value
			
			\foreach \k in {0,1} {
				\pgfmathtruncatemacro{\value}{\k * \Xdist};	\let\xPos=\value
				\pgfmathtruncatemacro{\value}{\k * 3};	\let\xNr=\value

				\ifthenelse{\k=0}{
					\let\NrString=i
					\let\NrStringTwo=j
				}{
					\let\NrString=j
					\let\NrStringTwo=i
				}

				\node[vertex2] (t-\xNr-0) at (\xPos,0 +\Ydist) {\ifthenelse{\k=0}{$u^t_\NrString$}{$u^t_\NrString$}};
				\node[vertex2] (vt2\k) at (\xPos,-1 +\Ydist) {$a^t_\NrString$};

				\foreach \i in {1,2}{
					\pgfmathtruncatemacro{\value}{\i + \k};	\let\sum=\value
					\ifthenelse{\sum=2}{
						\node[vertex2] (v-\i-1-\k) at (\i + \xPos - 1.5,-\yStretch + \Ydist - 1) {$r^\NrString_1$};
						\node[vertex2] (v-\i-2-\k) at (\i + \xPos - 1.5,-2*\yStretch + \Ydist - 1) {$r^\NrString_\NrStringTwo$};
						\node[vertex2] (v-\i-3-\k) at (\i + \xPos - 1.5,-3*\yStretch + \Ydist - 1) {$r^\NrString_n$};
					}{
						\node[vertex] (v-\i-1-\k) at (\i + \xPos - 1.5,-\yStretch + \Ydist - 1) {};
						\node[vertex] (v-\i-2-\k) at (\i + \xPos - 1.5,-2*\yStretch + \Ydist - 1) {};
						\node[vertex] (v-\i-3-\k) at (\i + \xPos - 1.5,-3*\yStretch + \Ydist - 1) {};
					}
				}

				\node[vertex2] (vb1\k) at (\xPos,\YPathDist + 1 + \Ydist) {$a^b_\NrString$};
				\node[vertex2] (t-\xNr-1) at (\xPos,\YPathDist + \Ydist) {$u^b_\NrString$};
				
				\foreach \i in {1,2}{
					\pgfmathtruncatemacro{\value}{\i + \k};	\let\sum=\value
					\foreach \j in {1,...,2}{
						\pgfmathtruncatemacro{\value}{1 + \j};	\let\nr=\value
						\path (v-\i-\j-\k) edge[dotted, thick] (v-\i-\nr-\k);
					}
					\path (v-\i-1-\k) edge (vt2\k);
					\path (v-\i-3-\k) edge (vb1\k);
				}
				\path (vb1\k) edge[very thick] (t-\xNr-1);
				\path (vt2\k) edge[very thick] (t-\xNr-0);
			}

			\node[vertex] (t-1-0) at (1 , \Ydist) {};
			\node[vertex] (t-2-0) at (\Xdist-1 , \Ydist) {};
			\path (t-1-0) edge[dotted, thick] node {$ j $} (t-0-0);
			\path (t-1-0) edge[dotted, thick] node {$ (j-i)y - i - j $}  (t-2-0);
			\path (t-3-0) edge[dotted, thick] node {$ i $} (t-2-0);
			
			\node[vertex] (w1) at (\distWToGadget,0.33*\YPathDist + \Ydist) {};
			\node[vertex] (w2) at (\Xdist - \distWToGadget,0.33*\YPathDist + \Ydist) {};
			\draw[-,decorate,thick] (\Xdist - \distWToGadget,0.33*\YPathDist + \Ydist - 1) -- (\distWToGadget,0.33*\YPathDist + \Ydist - 1);
			\node at (0.5 * \Xdist,0.33*\YPathDist + \Ydist - 1.5) {$(j-i-\frac{1}{2})y$};

			\node[above of=t-1-0, yshift=-6mm]  {$u^t_{i,j,1}$};
			\node[above of=t-2-0, yshift=-6mm]  {$u^t_{i,j,2}$};

			\node[below of=w1, yshift=6mm]  {$w^{i,j}_1$};
			\node[below of=w2, yshift=6mm]  {$w^{i,j}_2$};
			
			\path (t-0-1) edge[dotted, thick] node {$ (j-i)y $} (t-3-1);
			\path (w1) edge[dotted, thick] node {$ \M $} (w2);
			\path (t-1-0) edge[very thick] node {$ \TL $} (w1);
			\path (v-2-2-0) edge[very thick] node {$ \BL $} (w1);
			\path (t-2-0) edge[very thick] node {$ \TR $} (w2);
			\path (v-1-2-1) edge[very thick] node {$ \BR $} (w2);
		\end{tikzpicture}
	\end{center}
 	\caption{A schematic illustration of an edge-gadget~$g^E_{i,j}$ for the edge $\{v_i,v_j\}$. Dotted edges indicate paths of length more than one. The concrete length of these paths is indicated by the labels next to the edge. Bold edges indicate $y$-paths. The edge-gadget consists of the five parts denoted by~$\BL$ (bottom left), $\BR$ (bottom right), $\M$ (middle), $\TL$ (top right), and~$\TR$ (top right).}
	\label{fig:edge-gadget}
\end{figure*}

Finally, for all edges~$\{v_i,v_j\} \in E$ with $i<j$ insert an edge-gadget~$g^E_{i,j}$ into~$G'$ (see \autoref{fig:edge-gadget}): Add a path of length~$(j-i+\addFactor)y$ between the two right-vertices~$r^i_j$ and~$r^j_i$.
Denote with~$w^{i,j}_1$ the vertex on the path having distance~$y$ to~$r^i_j$ and denote with~$w^{i,j}_2$ the vertex on the path having distance~$y$ to~$r^j_i$. 
Furthermore, denote with~$u^t_{i,j,1}$ ($u^t_{i,j,2}$) the vertex in the top-line that lies between~$u^t_{i}$ and~$u^t_{j}$ and has distance~$j$ to~$u^t_{i}$ (distance~$i$ to~$u^t_{j}$).
Then  connect~$w^{i,j}_1$ by a $y$-path to~$u^t_{i,j,1}$ and also connect~$w^{i,j}_2$ by a $y$-path to~$u^t_{i,j,2}$. 
This completes the construction of~$G'$.

\section{Correctness of the Reduction} \label{sec:correctness}
Let~$(G=(V_1\cup V_2,E),h)$ be an instance of \BipDomSet and let $(G'=(V',E'),k)$ with $k=h+4$ be the corresponding instance of \MD that is constructed by the reduction above. Clearly, the reduction is polynomial-time computable and thus it remains to show that~$G$ has a dominating set of size~$h$ iff $G'$ has a metric basis of size~$k$. We first give an informal description of the basic ideas behind.

\subsection{Basic Ideas and Intuition}
First, observe that one has to choose at least one of the two degree-one vertices in the~$P_3$'s attached to each of the endpoints of the top- and bottom-line into any metric basis and a minimum-size metric basis would never take both. 
We shall show that $\{u^t_\ell,u^t_r,u^b_\ell,u^b_r\}$ separate each vertex pair in~$G'$ except the vertex pair $\{l^i_j,r^i_j\}$ for all $1\le i,j\le n$. 
Towards this the main observation is that a shortest path from a vertex in the skeletal structure to a vertex that is either in a vertex gadget or also in the skeletal structure would never enter an edge-gadget.
For example, traversing an edge-gadget~$g^E_{i,j}$ by entering it at~$u^t_{i,j,1}$ and leaving it at~$r^j_i$ gives a path of length~$(|j-i|+\addFactor)y$.
However, the path $u^t_i-u^t_j-a^t_j-r^j_i$ that follows the top-line is of length at most~$(|j-i|+1)y+n$ and, thus, is shorter (recall that~$\halfaddFactor y>2n+2$).
From this the separation of the vertices in the skeletal structure and most of the vertices in the vertex-gadgets can be deduced.
The reason why  $\{u^t_\ell,u^t_r,u^b_\ell,u^b_r\}$ cannot separate $\{l^i_j,r^i_j\}$ is that a shortest path starting in one of them has to enter a vertex-gadget~$g^V_i$ always via the anchors~$\{a^t_i,a^b_i\}$ and thus cannot distinguish between~$l^i_j$ and~$r^i_j$.


The fact that $\{u^t_\ell,u^t_r,u^b_\ell,u^b_r\}$ separate each vertex pair consisting only of edge-gadget vertices is far from being obvious and proving it requires extensive case distinctions (see \autoref{lem:twoEdgeGadgetSeparated}).
Moreover, we prove that all vertices in a metric basis of~$G'$ except $\{u^t_\ell,u^t_r,u^b_\ell,u^b_r\}$ are chosen from the vertex-gadgets and that the corresponding vertices form a dominating set in~$G$. Towards this it is crucial that the constant $\addFactor$ in the definition of edge-gadgets is between one and two: 
Clearly, taking~$r^i_1$ into a metric basis separates all pairs $\{l^i_j,r^i_j\}$ in its own gadget. 
The key point is that it separates also all pairs in $g^V_j$ if the edge-gadget $g^E_{i,j}$ exists: A path from $r^i_1$ to some $r^j_s$ and also to~$l^j_s$ via traversing~$g^E_{i,j}$ is of length at most $(|j-i|+\addFactor)y+3n$. 
The only ``alternative path'' from~$r^i_1$ to 
$r^j_s$ is by  leaving~$g^V_i$ via~$a^t_i$ following the top-line $u^t_i-u^t_j$, entering  $g^V_j$ via~$a^t_j$ and then taking the length-$s$ path to~$r^j_s$. In total this path has length at least $(|j-i|+2)y$ and thus traversing~$g^E_{i,j}$ is shorter.
Hence, because $\halfaddFactor y>2n+2$, the path traversing~$g^E_{i,j}$ is a shortest path. The idea behind is that leaving and entering vertex-gadgets via the anchors costs~$2y$ and traversing $g^E_{i,j}$ only costs $\addFactor y$ more than the top- or (bottom-)line path $u^t_i-u^t_j$ ($u^b_i-u^b_j$).
Moreover, $r^i_1$ only separates pairs in ``adjacent'' vertex-gadgets since a shortest path starting in~$r^i_1$ never traverses two edge-gadgets. This would cause at least two times the additional cost of $\addFactor y$ whereas leaving and entering~$g^V_i$ to and from the top-line only costs~$2y$.

We next give a formal proof of the correctness of the reduction. 

\subsection{General Observations and Additional Notation}

We first introduce some additional notation for edge-gadgets.

\emph{Notation:} For an edge-gadget $g^E_{i,j}$ the four vertices~$\{r^i_j, r^j_i, u^t_{i,j,1},u^t_{i,j,2}\}$ are called \emph{entrance-vertices} (see \autoref{fig:edge-gadget}). Moreover, we partition~$g^E_{i,j}$ into five parts: The $y$-path from~$w^{i,j}_1$ to~$r^i_j$ is the $\BL$- (bottom left) part, the $\TL$- (top left) part is the $y$-path between~$w^{i,j}_1$ and~$u^t_{i,j,1}$, the $\TR$- (top right) part is the $y$-path between~$w^{i,j}_2$ and~$u^t_{i,j,2}$, and the $\BR$- (bottom right) part is the $y$-path between~$w^{i,j}_2$ and~$r^j_i$. Part~$\M$ (middle) contains the remaining vertices, that is, the vertices between~$w^{i,j}_1$ and~$w^{i,j}_2$ including~$w^{i,j}_1$ and~$w^{i,j}_2$. 

A path \emph{enters} (\emph{leaves}) an edge-gadget~$g^E_{i,j}$ via a vertex~$v$ ($u$) if there are two consecutive vertices $v-u$ on it, such that both are contained in~$g^E_{i,j}$ and $v$ ($u$) is an entrance vertex of~$g^E_{i,j}$.
We say that an edge-gadget~$g^E_{i,j}$  is \emph{traversed} by a path~$P$ if it contains a subpath consisting only of vertices in~$g^E_{i,j}$ that starts with entering~$g^E_{i,j}$, contains the M-part, and ends with leaving~$g^E_{i,j}$.
Observe that without cycles there are only four different ways on how to traverse~$g^E_{i,j}$ and each is of length $(|j-i|+\addFactor)y$. 
The following observations are straightforward.
\begin{obs}\label{obs:traversing-edge-gadget}%
A path that enters and afterwards leaves an edge-gadget without traversing it is more than $\halfaddFactor y$ longer than a shortest path with the same endpoints.
\end{obs}
{\begin{proof}
Let~$P$ be a path that starts with entering an edge-gadget~$g^E_{i,j}$ and ends with leaving it but does not traverse it. 
If~$P$ does not contain any cycle, then it is equal either to~$u^t_{i,j,1}-w^{i,j}_1-r^i_j$ or to $u^t_{i,j,2}-w^{i,j}_2-r^j_i$. 
Thus $P$ is of length~$2y$. 
However, the paths $u^t_{i,j,1}-u^t_i-a^t_i-r^i_j$ and $u^t_{i,j,2}-u^t_j-a^t_j-r^j_i$ are both of length at most~$y+2n$.
Thus~$P$ is by $y - 2n > \halfaddFactor y$ longer. \oldqed
\end{proof}}
Next we prove that any shortest path traverses at most one edge-gadget.

\begin{obs}\label{lem:traversing-edge-gadgets}
If a path is at most $\addFactorMinusOne y$ longer than a shortest path with the same endpoints, then it traverses at most one edge-gadget.
\end{obs}
{\begin{proof}
Assume that there is a path~$P$ in $G'=(V',E')$ that  starts with traversing an edge-gadget~$g^E_{i,j}$ and ends with traversing~$g^E_{i',j'}$.  Then $P$ has length at least~$(j-i+\addFactor+ j'-i'+\addFactor)y$.
We assume that it starts traversing~$g^E_{i,j}$ either in~$r^i_j$ or in~$u^t_{i,j,1}$ (the other entrance vertices are completely symmetric). Then~$P$ ends either in $r^{i'}_{j'}$ or~$u^t_{i',j',1}$.  
However, the path $r^i_j-a^t_i-u^t_i-\ldots-u^t_{i'}$ and also $u^t_{i,j,1}-u^t_{i'}$ both have length at most~$j+(1+i'-i)y$ and $u^t_{i'}$ has distance at most~$2y$ to $u^t_{i',j',1}$ and~$r^{i'}_{j'}$. 
Hence, there is path that avoids traversing~$g^E_{i,j}$ with distance at most $j+(1+i'-i)y<(j-i+\addFactor+ j'-i'+\addFactor)y$, implying that~$P$ is more than~$\addFactorMinusOne y$ longer than a shortest path.%
\oldqed%
\end{proof}}%
We next prove that for all vertices in~$G'$ except for the vertices contained in an edge-gadget it holds that a shortest path to a vertex on the top- or bottom-line never contains a vertex in an edge-gadget.

\begin{lemma} \label{lem:skel-distances}
In~$G'$ the following  holds:
	
\begin{enumerate}[i)]
	\item \label{prop:dist-on-one-line}For all $1\le i<j\le n$ the path along the top-line (bottom-line) is a shortest path from $u^t_i$ to $u^t_j$ ($u^b_i$ to $u^b_j$). It has length $(j-i)y$.

\item \label{prop:dist-on-different-lines} For all $1\le i,j\le n$ the following paths of length $(|j-i|+2)y+n+1$ are for all $\min\{i,j\}\le s\le \max\{i,j\}$ the only shortest paths between $u^t_i$ and $u^b_j$:
	\begin{align*}
u^t_i-u^t_{i+1}-\ldots-u^t_s-a^t_s-r^s_1-r^s_2-\ldots\\
-r^s_n-a^b_s-u^b_s-u^b_{s+1}-u^b_j\\
u^t_i-u^t_{i+1}-\ldots-u^t_s-a^t_s-l^s_1-l^s_2-\ldots\\
-l^s_n-a^b_s-u^b_s-u^b_{s+1}-u^b_j.
\end{align*}%
\end{enumerate}
\end{lemma}
{\begin{proof}
 \emph{[Proof of i):]} We prove the claim for $i=1$ and $j=n$ on the top-line. As the vertices for all other choices of $i$ and~$j$  also lie on the top-line, this implies the correctness in all other cases.  \autoref{lem:skel-distances}(\ref{prop:dist-on-one-line}) can be analogously proven for the bottom-line. 

Assume that there is a shortest path~$P$ from $u^t_1$ to~$u^t_n$ that does not follow the top-line. We first show that~$P$ traverses at least one edge-gadget: Because the distances on the top- and the bottom-line are completely symmetric, a shortest path never starts on the top-line enters at some point the bottom-line and then later on re-enters the top-line. Hence, when leaving the top-line a shortest path enters an edge-gadget and, by \autoref{obs:traversing-edge-gadget}, it traverses it, say~$g^E_{i,j}$.

Since $P$ is a shortest path and traverses by~\autoref{lem:traversing-edge-gadgets} only~$g^E_{i,j}$, it follows that~$P$ enters~$g^E_{i,j}$ via $u^t_{i,j,1}$ and leaves it via~$u^t_{i,j,2}$. This subpath in~$P$ is of length $(j-i+\addFactor)y$. Contradictorily, the path from $u^t_{i,j,1}$ to~$u^t_{i,j,2}$ along the top-line is of length less than $(j-i)y$.

\emph{[Proof of ii):]} As the other case can be proven completely analogously, we prove the claim only for the case $i\le j$. Furthermore, for any choice of $i\le s\le j$ the corresponding path has length $(|j-i|+2)y+n+1$. Thus it remains to prove that every other path is longer.

Suppose towards a contradiction that there is a path~$P'$ from $u^t_i$ to~$u^b_j$ of length at most $(j-i+2)y+n+1$ that is different from the paths described in \autoref{lem:skel-distances}(\ref{prop:dist-on-different-lines}). By \autoref{lem:skel-distances}(\ref{prop:dist-on-one-line}) it holds that if there are two vertices in~$P'$ that lie on the top-line (bottom-line), then all vertices on the subpath between them also lie on the top-line (bottom-line, resp.). Thus, there are two vertices $u^t_\alpha$ and $u^b_\beta$ that are both in~$P'$ but $u^t_{\alpha+1},u^b_{\beta-1}\notin P'$. 
If on the subpath from $u^t_\alpha$ to $u^b_\beta$ no edge-gadget is used, then $\beta=\alpha$ and the path~$P'$ is identical to the path described by \autoref{lem:skel-distances}(\ref{prop:dist-on-one-line}) for $s=\alpha$.
	
\looseness=-1 Thus~$P'$ traverses exactly one edge-gadget (\autoref{lem:traversing-edge-gadgets}), say~$g^E_{\alpha,\beta}$. Hence the subpath in~$P'$ from~$u^t_\alpha$ via $g^E_{\alpha,\beta}$ to~$u^b_\beta$ has length at least $(\beta-\alpha+\addFactor+1)y$. Contradictorily, the path $u^t_\alpha-a^t_\alpha-a^b_\alpha-u^b_\alpha-u^b_{\alpha+1}-\ldots-u^b_\beta$ is of length at most $(2+\beta-\alpha)y+2n$ and thus is by at least~$\halfaddFactor y$ shorter.
\oldqed%
\end{proof}}
\autoref{obs:traversing-edge-gadget}, \autoref{lem:traversing-edge-gadgets}, and \autoref{lem:skel-distances} together with the following proposition is all what we need to prove the correctness of our reduction.
\begin{proposition}\label{prop:all-but-vertex-is-separated}
  The four vertices $\{u^t_\ell,u^t_r,u^b_\ell,u^b_r\}$ separate all vertices in~$G'$ except the pairs $\{l^i_j,r^i_j\}$ for all $1\le i,j\le n$.
 \end{proposition}
The next subsection is dedicated to prove \autoref{prop:all-but-vertex-is-separated}. Based on it, in \autoref{subsec:correctness} we give a formal correctness proof of the reduction.

\subsection{Proof of \autoref{prop:all-but-vertex-is-separated}}

The major work in proving \autoref{prop:all-but-vertex-is-separated} is to show that the vertices in the edge-gadgets are separated.
To this end, we first show which entrance vertices are used by shortest paths starting in $\{u^t_\ell,u^t_r,u^b_\ell,u^b_r\}$ and ending in an edge-gadget vertex.

\begin{obs}\label{obs:endpoints-entering-edge-gadgets}
 Let $e$ be a vertex in an edge-gadget~$g^E_{i,j}$. If~$e$ is contained in the $\TL$-, $\BL$- ($\TR$-, $\BR$-,) or $\M$-part, then all shortest paths from $u^t_\ell$ or $u^b_\ell$ ($u^t_r$ or $u^b_r$) to~$e$ enter $g^E_{i,j}$ either via~$u^t_{i,j,1}$ ($u^t_{i,j,2}$) or~$r^i_j$ ($r^j_i$).
\end{obs}
\begin{proof}
	We prove the claim for shortest paths starting in~$u^t_\ell$ or~$u^b_\ell$. 
	The other case is completely symmetric. 
	A path from $u^t_\ell$ or from $u^b_\ell$ to~$e$ that neither enters~$g^E_{i,j}$ via~$r^i_j$ nor via~$u^t_{i,j,1}$ has to traverse either the $\TR$ or $\BR$-part till~$w^{i,j}_2$. 
	With this requirement the shortest paths are $u^t_\ell-u^t_{j-1}-u^t_{i,j,2}-w^{i,j}_2$ with length $2+jy-i$ and $u^b_\ell-u^b_{j}-a^b_j-r^j_i-w^{i,j}_2$ with length $2+(j+1)y+(n-i+1)$. 
	However, each of the paths $u^t_\ell-u^t_i-u^t_{i,j,1}-w^{i,j}_1-w^{i,j}_2$ with length $2+j+(j-\addFactorMinusOne)y$ and $u^b_\ell-u^b_i-a^b_i-r^i_j-w^{i,j}_1-w^{i,j}_2$ with length $2+(n-j+1)+(j+\addFactorMinusOne)y$ are at least by $\halfaddFactor y$ shorter, respectively.\oldqed
\end{proof}
\begin{lemma}\label{lem:twoEdgeGadgetSeparated}
The four vertices $\{u^t_\ell,u^t_r,u^b_\ell,u^b_r\}$ separate each vertex-pair consisting of two edge-gadget vertices.	
\end{lemma}
\appendixproof{\autoref{lem:twoEdgeGadgetSeparated}}
{\begin{proof}
Let~$u$ be a vertex in the edge-gadget~$g^E_{i,j}$ with $i<j$, and let~$v$ be a vertex in the edge-gadget~$g^E_{i',j'}$ with $i'<j'$. By our vertex numbering it follows that $i<j'$ and $i'<j$. We shall show that~$u$ and~$v$ are separated by~$u^t_1,u^t_n,u^b_1$, or $u^b_r$, implying that they are also separated by the corresponding degree-one vertices $\{u^t_\ell,u^t_r,u^b_\ell,u^b_r\}$.
	
Recall that \autoref{lem:skel-distances} provides the distance between any vertex on the top- or bottom-line  to any other vertex either contained in a vertex-gadget or on top- or bottom-line. The following distances will be frequently used:
\begin{align}
		\dist(r^i_j, u^t_1) & = iy+j							\label{eq:dist-rij-to-ut1} \\
		\dist(r^i_j, u^t_n) & = (n-i+1) y + j 					\label{eq:dist-rij-to-utn}\\
		\dist(r^i_j, u^b_1) & = iy + n - j + 1 					\label{eq:dist-rij-to-ub1}\\
		\dist(r^i_j, u^b_n) & = (n-i+1)y + n - j + 1 			\label{eq:dist-rij-to-ubn}\\
		\dist(u^t_{i,j,1}, u^t_1) & = (i-1)y + j 				\label{eq:dist-utij1-to-ut1}\\
		\dist(u^t_{i,j,1}, u^t_n) & = (n-i)y - j 				\label{eq:dist-utij1-to-utn}\\
		\dist(u^t_{i,j,1}, u^b_n) & = (n-i+2)y - j + n + 1 		\label{eq:dist-utij1-to-ubn}\\
		\dist(u^t_{i,j,2}, u^t_1) & = (j-1)y - i 				\label{eq:dist-utij2-to-ut1}\\
		\dist(u^t_{i,j,2}, u^t_n) & = (n-j)y + i 				\label{eq:dist-utij2-to-utn}\\
		\dist(u^t_{i,j,2}, u^b_n) & = (n-j+2)y + i + n + 1		\label{eq:dist-utij2-to-ubn}
	\end{align}
	We prove \autoref{lem:twoEdgeGadgetSeparated} by several case distinctions.
Therein, the following five claims are helpful to simplify the argumentation (the proofs are separate subsections in the appendix).

\medskip \noindent {\bf Claim 1:} $\{u,v\}$ are separated if $i=i'$ and~$j=j'$.
\appendixproofT{Claim~1 in \autoref{lem:twoEdgeGadgetSeparated}}
{\begin{proof}[Claim 1:]
		It follows that $u$ and~$v$ are contained in~$g^E_{i,j}$ and we make a case distinction on in which part they lie.

		\emph{{\bf Case 1:} At least one of $\{u,v\}$ is contained in the $\M$-part.} 
		Suppose $u$ is contained in the $\M$-part, then by \autoref{obs:endpoints-entering-edge-gadgets} a shortest path from~$u^t_\ell$ ($u^b_\ell$) to $u$ contains the subpath $u^t_{i,j,1}-w^{i,j}_1$ ($r^i_j-w^{i,j}_1$) and from there traverses the M-part till~$u$. 
		Clearly, if $v$ is also contained in the $\M$-part, then the same holds for~$v$ and thus $\dist(u^t_\ell,v)\neq \dist(u^t_\ell,u)$. 
		If $v$ is either contained in the $\TL$- or $\BL$-part, then $v$ is contained either on a shortest path from $u^t_\ell$ to~$u$ or on a shortest path from $u^b_\ell$ to~$u$. 
		The remaining subcase where $v$ is contained in the $\TR$- or $\BR$-part can be proven analogously by interchanging the role of $\{u^t_\ell,u^b_\ell\}$ by $\{u^t_r,u^b_r\}$.

		\emph{{\bf Case 2:} $u,v$ are both either in the $\TL$-, $\BL$-, $\TR$-, or $\BR$-part.}
		Observe that if $u$ and~$v$ are both on the $\TL$- ($\BL$-, $\TR$-, $\BR$-) part, then they are separated by $u^t_\ell$ ($u^b_\ell,u^t_r,u^b_r$) since there is a shortest path from $u^t_\ell$ or~$u^b_\ell$ ($u^t_r,u^b_r$) to~$w^{i,j}_1$ ($w^{i,j}_2$) that contains~$v$ and~$u$.

		\emph{{\bf Case 3:} $u$ is contained in the $\TL$- or $\BL$-part and $v$ is contained in the $\TR$- or $\BR$-part.} 
		Recall that by our vertex numbering it holds that~$j\ge i+3$. In this case, since $u^t_i$ has distance less than $2y$ to all vertices on the $\TL$- and $\BL$-part, it follows that $\dist(u^t_1,u)<(i+1)y$. 
		Additionally, a shortest path from~$u^t_1$ to~$v$ has to contain either~$u^t_{i,j,2}$, $r^j_i$, or $w^{i,j}_2$. 
		Since $\dist(u^t_1,u^t_{i,j,2}) \overset{(\ref{eq:dist-utij2-to-ut1})}{=} (j-1)y-i\ge (i+2)y$, $\dist(u^t_1,w^{i,j}_2)>(i-1+j-i+\addFactorMinusOne)y\ge (i+2+\addFactorMinusOne)y$, and $\dist(u^t_1,r^j_i)>\dist(u^t_1,u^t_{i,j,2})$ it follows that~$u$ and~$v$ are separated by~$u^t_1$.

		\looseness=-1 \emph{{\bf Case 4:} $u$ is contained in the $\TL$- ($\TR$-) part and $v$ is contained in the $\BL$- ($\BR$-) part.} 
		We prove the claim in case of $u\in \TL$ and $v\in \BL$. The other case is completely symmetric.
		If the shortest path from~$v$ to~$u^t_1$ goes via~$u^t_{i,j,1}$, then~$u$ lies on this shortest path and, thus, $u$ and~$v$ are separated.
		If the shortest path from~$v$ to~$u^t_1$ goes via~$r^i_j$, then observe that~$\dist(u^t_1, v) > \dist(u^t_1, r^i_j) \overset{(\ref{eq:dist-rij-to-ut1})}{=} iy+j$.
		Furthermore, from \autoref{eq:dist-utij1-to-ut1} it follows that~$\dist(u^t_1, u) \le (i-1)y + j + y = iy+j$. 		\oldqed
	\end{proof}}

	\medskip \noindent {\bf Claim~2:} $\{u,v\}$ are separated if~$u \in \TL \cup \BL$ and $i<i'$. Symmetrically, $\{u,v\}$ are separated if $u\in \TR\cup \BR$ and $j>j'$.
\appendixproofT{Claim~2 in \autoref{lem:twoEdgeGadgetSeparated}}	
	{\begin{proof}[Claim~2:]
			We prove that if $u\in \TL\cup\BL$ and $i<i'$, then the vertex~$u^t_1$ or~$u^b_1$ separate $\{u,v\}$:
			It follows that $\dist(u^t_1,u,u^b_1)$
			\begin{align*}
				& \le \dist(u^t_1,u^t_{i,j,1},u)+\dist(u^b_1,r^i_j,u)\\
				& \le \dist(u^t_1,u^t_{i,j,1})+\dist(u^b_1,r^i_j)  + 2y \\
				& \overset{(\ref{eq:dist-rij-to-ub1},\ref{eq:dist-utij1-to-ut1})}{=} (iy + n - j + 1) + ((i-1)y + j) + 2y \\ 
				& = (2i+1)y + n + 1.
			\end{align*}
			Furthermore, from \autoref{obs:endpoints-entering-edge-gadgets} it follows that~$\dist(u^t_1,v,u^b_1)$
			\begin{align*}
				& > \dist(u^t_1,u^t_{i',j',1})+\dist(u^b_1,r^{i'}_{j'})\\
				& \overset{(\ref{eq:dist-rij-to-ub1},\ref{eq:dist-utij1-to-ut1})}{=} ((i'-1)y + j')+(i'y + n - j' + 1)\\ 
				& = (2i'-1)y + n + 1 \ge (2i+1)y + n + 1.
			\end{align*}
			Hence, $u^b_1$ and~$u^t_1$ separate $\{u,v\}$.
			The symmetric case can be proven analogously be interchanging the role of $\{u^t_1,u^b_1\}$ by~$\{u^t_n,u^b_n\}$.
		\oldqed
	\end{proof}}

\noindent {\bf Claim~3:} $\{u,v\}$ are separated if $i+j\neq i'+j'$
and either
\begin{enumerate}[i)]
  \item $\dist(u^t_{i,j,1},u,u^t_{i,j,2}) = (j-i+\addFactor)y$ and $\dist(u^t_{i',j',1},v,u^t_{i',j',2}) = (j'-i'+\addFactor)y$, or
\item $\dist(r^i_j,u,r^j_i) = (j-i+\addFactor)y$ and $\dist(r^{i'}_{j'},v,r^{j'}_{i'}) = (j'-i'+\addFactor)y$.
\end{enumerate}
\appendixproofT{Claim~3 in \autoref{lem:twoEdgeGadgetSeparated}}	
{\begin{proof}[Claim~3:]
	By the requirements it is ensured that a path from
\begin{enumerate}[i)]
  \item $u^t_1$ to~$u^t_n$ via $u$ ($v$) enters $g^E_{i,j}$ ($g^E_{i',j'}$) via $u^t_{i,j,1}$ ($u^t_{i',j',1}$), traverses it, and leaves it via~$u^t_{i,j,2}$ ($u^t_{i',j',2}$), or
\item $u^b_1$ to $u^b_n$ via $u$ ($v$) enters $g^E_{i,j}$ ($g^E_{i',j'}$) via $r^i_j$ ($r^{i'}_{j'}$), traverses it, and leaves it via~$r^j_i$ ($r^{j'}_{i'}$).
\end{enumerate}
In case of i) it holds that $\dist(u^t_1,u,u^t_n)$
			\begin{align*}
				& = \dist(u^t_1,u^t_{i,j,1},u) + \dist(u,u^t_{i,j,2},u^t_n) \\
									& \overset{(\ref{eq:dist-utij1-to-ut1},\ref{eq:dist-utij2-to-utn})}{=} ((i-1)y + j ) + ((n-j)y + i) + (j-i+\addFactor)y \\
									& = i+j+(n+\addFactor-1)y.
			\end{align*}
			Symmetrically, $\dist(u^t_1,v,u^t_n)	= i'+j'+(n+\addFactor-1)y$. Since~$i + j \neq i'+j'$ it follows that~$u$ and~$v$ are separated by~$u^t_1$ and~$u^t_n$.

Now, assume that ii) holds, then $\dist(u^b_1,u,u^b_n)$
\begin{align*}
 = {} & \dist(u^b_1,r^i_j,u)+\dist(u,r^j_i,u^b_n)\\
	\overset{(\ref{eq:dist-rij-to-ub1},\ref{eq:dist-rij-to-ubn})}{=} {} & iy+n-j+1+n-i+1\\&+(n-j+1)y+(j-i+\addFactor)y\\
	= {} & (n+1+\addFactor)y+2n-j-i+2
\end{align*}
and, symmetrically, $\dist(u^b_1,v,u^b_n)=(n+1+\addFactor)y+2n-j'-i'+2$. Since $i+j\neq i'+j'$ it follows that~$u$ and~$v$ are separated.
		\oldqed
	\end{proof}}

	\noindent {\bf Claim~4:} $\{u,v\}$ are separated if $u \in \M$ and $v \in \TL'\cup \TR'$.
\appendixproofT{Claim~4 in \autoref{lem:twoEdgeGadgetSeparated}}	
{\begin{proof}[Claim~4:] We prove the claim for $v\in \TL'$. The case with $v\in \TR'$ follows from the symmetry of the construction.
			From \autoref{obs:endpoints-entering-edge-gadgets} follows~$\dist(u^t_1,v) = \dist(u^t_1, u^t_{i',j',1},v) \overset{(\ref{eq:dist-utij1-to-ut1})}{=} (i'-1)y+j'+x$ with~$x = \dist(u^t_{i',j',1},v) < y$.
			Furthermore, $\dist(u^t_1,u) = \dist(u^t_1,u^t_{i,j,1}, w^{i,j}_1, u) = (i-1)y + j + y + \dist(w^{i,j}_1, u) = iy + j + \dist(w^{i,j}_1, u)$.
			Assuming that~$\dist(u^t_1,v) = \dist(u^t_1,u)$ (otherwise~$u$ and~$v$ are separated) we have~$(i'-1)y+j'+x = iy + j + \dist(w^{i,j}_1, u)$ and, hence, $\dist(w^{i,j}_1, u) = (i'-i-1)y + j' - j + x$.
			Thus,~$i' \ge i$.
			From this together with \autoref{obs:endpoints-entering-edge-gadgets} it follows that $\dist(u^b_1,u)$
			\begin{align*}
				 	& = \dist(u^b_1,r^i_j, w^{i,j}_1, u)\\ 
				& \overset{(\ref{eq:dist-rij-to-ub1})}{=} iy+n-j+1 + y + (i'-i-1)y + j' -j +x \\
				& = i'y + n +j' -2j + 1 + x.
			\end{align*}
			Furthermore, it follows that $\dist(u^b_1,v) = \dist(u^b_1,r^{i'}_{j'},v) \overset{(\ref{eq:dist-rij-to-ub1})}{=} i'y+n-j'+1 + \dist(r^{i'}_{j'},v)$.
			By \autoref{obs:endpoints-entering-edge-gadgets} $\dist(r^{i'}_{j'},v)$
			\begin{align*} &= \min\{ \dist(r^{i'}_{j'},u^t_{i',j',1},v) , \dist(r^{i'}_{j'},w^{i',j'}_1,v) \} \\
			&= \min\{y + 2j' + x, 2y - x\}.\end{align*}
			Hence, $\dist(u^b_1,v) \overset{(\ref{eq:dist-rij-to-ub1})}{=} i'y + n -j' + 1 + \min\{y + 2j' + x, 2y - x\}$.
			Assuming~$\dist(u^b_1,u) = \dist(u^b_1,v)$ it follows that
			\begin{align*}
				i'y + n +j' -2j + 1 + & x \\ = i'y + n -j' + & 1 + \min\{y + 2j' + x, 2y - x\} 
\end{align*} and thus $2j'-2j + x  = \min\{y + 2j' + x, 2y - x\}$.
			
 			Since the case $2j'-2j + x = y + x + 2j'$ yields a contradiction ($y=-2j$), it follows that $2j'-2j=2y-2x$. However, $x<y$ implies $j'>j$ and thus $i\le i'<j<j'$. Since $x=y-j'+j$, 
 $\dist(u^t_{i',j',1},v,u^t_{i',j',2})=\dist(u^t_{i',j',1},v)+\dist(v,w^{i',j'}_1,w^{i',j'}_2,u^t_{i',j',2})$, implying that all preconditions of Claim~3 are fulfilled and thus~$u$ and~$v$ are separated.\oldqed
	\end{proof}}

	\noindent {\bf Claim~5:} $\{u,v\}$ are separated if $u \in \M$ and $v \in \BL'\cup \BR'$.
	\appendixproofT{Claim~5 in \autoref{lem:twoEdgeGadgetSeparated}}
	{\begin{proof}[Claim~5:]
			We prove the claim for $v\in \BL'$. The case with $v\in \BR'$ follows from the symmetry of the construction.
			From \autoref{obs:endpoints-entering-edge-gadgets} it follows that~$\dist(u^b_1,v) = \dist(u^b_1, r^{i'}_{j'},v) \overset{(\ref{eq:dist-rij-to-ub1})}{=} i'y+n-j'+1+x$ with~$x = \dist(r^{i'}_{j'},v) < y$.
			Furthermore, ~$\dist(u^b_1,u) = \dist(u^b_1,r^i_j, w^{i,j}_1, u) \overset{(\ref{eq:dist-rij-to-ub1})}{=} iy + n - j + 1 + y + \dist(w^{i,j}_1, u) = (i+1)y + n - j + 1 + \dist(w^{i,j}_1, u)$.
			Assuming that~$\dist(u^b_1,v) = \dist(u^b_1,u)$ (otherwise~$u$ and~$v$ are separated) we have~$i'y+n-j'+1+x = (i+1)y+n-j+1 + \dist(w^{i,j}_1, u)$ and, hence, $\dist(w^{i,j}_1, u) = (i'-i-1)y + j - j' + x$.
			Thus,~$i' \ge i$.
			This together with \autoref{obs:endpoints-entering-edge-gadgets} implies $\dist(u^t_1,u)$
			\begin{align*}
				 	& = \dist(u^t_1,u^t_{i,j,1}, w^{i,j}_1)+\dist(w^{i,j}_1,u) \\ 
								& \overset{(\ref{eq:dist-utij1-to-ut1})}{=} (i-1)y + j + y + (i'-i-1)y + j -j' +x \\
								& = (i'-1)y -j' + 2j + x.
			\end{align*}
			In addition, $\dist(u^t_1,v) = \dist(u^t_1,u^t_{i'},v) = (i'-1)y + \dist(u^t_{i'},v)$
			and
	\begin{align*}\dist(u^t_{i'},v) &= \min\{ \dist(u^t_{i'},w^{i',j'}_1,v) , \dist(u^t_{i'}, r^{i'}_{j'},v) \}\\& =  y + j' + \min\{y - x, x\}.\end{align*}
			Hence, $\dist(u^t_1,v) = i'y  + j' + \min\{x, y - x\}$.
			Assuming $\dist(u^t_1,u) = \dist(u^t_1,v)$ (otherwise~$u$ and~$v$ are separated) implies
			\begin{align*}
				(i'-1)y -j' + 2j + x & = i'y + j' + \min\{x, y - x\} \\
				2j-2j' + x & = y + \min\{x, y - x\}.
			\end{align*}
			\looseness=-1 This gives that either~$2j-2j' + x = y + x$ or~$2j-2j' + x = 2y - x$.
			In the first case this gives~$y = 2j - 2j'$, contradicting~$\halfaddFactor y>2n$.
			The second case gives~$x = y + j' - j$.
			Since~$x < y$ it follows that~$j > j'$ and, thus, $i \le i' < j' < j$.
%
			However, since~$\dist(r^{i'}_{j'},v) = x = y + j' - j$, it follows that
			\begin{align*}
				\dist(u^t_n,v) 	& = \dist(u^b_1,u^t_{i',j',2},w^{i',j'}_2,w^{i',j'}_1,v) \\
								& \overset{(\ref{eq:dist-utij2-to-utn})}{=} (n-j')y+ i' + (j'-i'+\addFactorMinusOne)y + j-j' \\
								& = (n-i'+\addFactorMinusOne)y+i'-j' + j.
			\end{align*}
			Furthermore, 
			\begin{align*}
				\dist(u^t_n,u) 	= {} & \dist(u^b_n,u^t_{i,j,2},w^{i,j}_2,w^{i,j}_1) - \dist(w^{i,j}_1,u) \\
								\overset{(\ref{eq:dist-utij2-to-utn})}{=} {} & (n-j)y+i + (j-i+\addFactorMinusOne)y\\& - ((i'-i-1)y + j - j' + x) \\
								= {} & (n - i + \addFactorMinusOne)y + i\\
								& - ( (i' - i-1)y + j - j' + x )\\
								= {} & (n - i'+ \addFactor)y + i - j + j' - x \\
								= {} & (n - i'+ \addFactorMinusOne)y   + i. 
			\end{align*}
			Hence, 
			\begin{align*}
				\dist(u^b_n,u) - \dist(u^b_n,v)
								& = i  - (i' - j' + j) \\
								& = i - i' + j' - j.
			\end{align*}
			Recall that~$i \le i'$ and~$j' < j$.
			Thus, $\dist(u^b_n,u) - \dist(u^b_n,v) < 0$ and, hence, $u$ and~$v$ are separated by~$u^b_n$.
		\oldqed
	\end{proof}}
 
We now prove \autoref{lem:twoEdgeGadgetSeparated} by a case distinction on how the indices $i,i',j,$ and~$j'$ are related to each other.
Without loss of generality, we assume that $i\le i'$. 
Moreover, by Claim~1 either $i \neq i'$ or $j\neq j'$. 
We first prove the case with $j=j'$ and $i \neq i'$ (Case~1). 
The case where $i=i'$ and $j \neq j'$ is omitted because it can be proven completely analogously. 
Hence, the remaining cases are~$i<i' < j < j'$ (Case~2) and $i < i' < j' < j$ (Case~3). Note that in all these cases, by Claim~2 we may assume that $u\notin \TL\cup \BL$.

\smallskip
\noindent \textbf{Case 1 $i < i' < j = j'$:}\\
If $u\in M$, then Claim~3, 4, and 5 prove that $\{u,v\}$ are separated. It remains to consider $u\in \TR\cup \BR$.

			\noindent\emph{Subcase 1:} $u \in \TR$.\\
			If $\dist(u^t_n, u) = \dist(u^t_n, v)$, then $\dist(u^t_n, u) = \dist(u^t_n, u^t_{i,j,2}, u) \overset{(\ref{eq:dist-utij2-to-utn})}{<} (n-j)y + i + y$.
			Since~$i<i'$ and it follows from Equations~\ref{eq:dist-rij-to-utn}, \ref{eq:dist-utij1-to-utn}, and \ref{eq:dist-utij2-to-utn} that~$v \in \TR'$.
			Thus,~$\dist(u,u^t_{i,j,2}) = i' - i + \dist(v,u^t_{i',j',2})$.
			Let~$x = \dist(u,u^t_{i,j,2})$, then~$\dist(v,u^t_{i',j',2}) = x + i -i'$.

			\emph{Subcase 1.1:} $x < \frac{2}{3}y$.\\
			Then~$\dist(u^t_1,u) = \dist(u^t_1, u^t_{i,j,2}) + x$ and~$\dist(u^t_1,v) = \dist(u^t_1, u^t_{i',j',2}) + x + i - i'$.
			Since~$i < i'$ it follows from \autoref{eq:dist-utij2-to-ut1} that $\dist(u^t_1, u^t_{i,j,2}) > \dist(u^t_1, u^t_{i',j',2})$.
			Hence~$\dist(u^t_1,u) > \dist(u^t_1,v)$.

			\emph{Subcase 1.2:} $x \ge \frac{2}{3}y$.\\
			Then~$\dist(u,u^b_1) = \dist(u,r^i_j,u^b_1) \overset{(\ref{eq:dist-rij-to-ub1})}{=} iy+n-j+1 + (j-i+\addFactor)y - x = (j+\addFactor)y + n - j + 1 - x$.
			Furthermore, $\dist(v,u^b_1) = \dist(v,r^{i'}_{j'},u^b_1) \overset{(\ref{eq:dist-rij-to-ub1})}{=} i'y+n-j'+1 + (j'-i'+\addFactor)y - x - i + i' = (j+\addFactor)y + n - j + 1 - x - i + i'$.
			Since~$i\neq i'$ it follows that~$u$ and~$v$ are separated.

			\noindent\emph{Subcase 2:} $u \in \BR$.\\
			This subcase is analogous to the previous one:
			Assume that~$\dist(u^b_n, u) = \dist(u^b_n, v)$.
			Then~$\dist(u^b_n, u) = \dist(u^b_n, r^j_i , u) \overset{(\ref{eq:dist-rij-to-ubn})}{<} (n-j+1)y + n - i + 1 + y$.
			Since~$i<i'$ it follows from Equations~\ref{eq:dist-rij-to-ubn}, \ref{eq:dist-utij1-to-ubn}, and \ref{eq:dist-utij2-to-ubn} that~$v \in \BR'$.
			Thus,~$\dist(u,r^j_i) = i - i' + \dist(v,r^{j'}_{i'})$.
			Denoting with~$x = \dist(u,r^j_i)$ it follows that~$\dist(v,r^{j'}_{i'}) = x - i + i'$.

			\emph{Subcase 2.1:} $x < \frac{2}{3}y$.\\
			Then~$\dist(u^t_n,u) = \dist(u^t_n, r^j_i) + x$ and~$\dist(u^t_n,v) = \dist(u^t_n, r^{j'}_{i'}) + x - i + i'$.
			Since~$i < i'$ it follows from \autoref{eq:dist-rij-to-utn} that $\dist(u^t_n, r^j_i) < \dist(u^t_n, r^{j'}_{i'})$.
			Hence~$\dist(u^t_n,u) < \dist(u^t_n,v)$.

			\emph{Subcase 2.2:} $x \ge \frac{2}{3}y$.\\
			Then~$\dist(u,u^t_1) = \dist(u,u^t_{i,j,1},u^t_1) \overset{(\ref{eq:dist-utij1-to-ut1})}{=} (i-1)y+j + (j-i+\addFactor)y - x = (j+\frac{1}{2})y+j-x$.
			Furthermore, $\dist(v,u^t_1) = \dist(v,u^t_{i',j',1},u^t_1) \overset{(\ref{eq:dist-utij1-to-ut1})}{=} (i'-1)y+j' + (j'-i'+\addFactor)y - x + i - i' = (j+\frac{1}{2})y+j-x + i - i'$.
			Since~$i\neq i'$ it follows that~$u$ and~$v$ are separated.

		\smallskip
		\noindent \textbf{Case 2 $i < i' < j < j'$:}\\
		If $u\in M$, then Claim~3, 4, and 5 prove that $\{u,v\}$ are separated. It remains to consider $u\in \TR\cup \BR$.
 
\noindent \emph{Subcase 1:} $u \in \TR$.\\
			It follows from the Claims~2, 3, 4, and 5 that $v \in \TL'\cup \BL'$.
			If~$\dist(u^t_n, u) = \dist(u^t_n, v)$, then~$\dist(u^t_n, u) \overset{(\ref{eq:dist-utij2-to-utn})}{=} (n-j)y + i + x$ where~$x = \dist(u,u^t_{i,j,2}) < y$.
			It follows that~$v \in \TL'$ and~$\dist(u^t_n,v) = \dist(u^t_n,u^t_{i',j',1},v) \overset{(\ref{eq:dist-utij1-to-utn})}{=} (n-i')y - j' + \dist(u^t_{i',j',1},v)$.
			Assuming~$\dist(u^t_n, u) = \dist(u^t_n, v)$, we have~$j = i'+1$.
			Hence, 
			\begin{align*}
				(n-j)y + i + x & = (n-j+1)y - j' + \dist(u^t_{i',j',1},v) \\
				\dist(u^t_{i',j',1},v) & = x + i + j' - y.
			\end{align*}
			Since~$x < y$ it follows that~$\dist(u^t_{i',j',1},v) < i + j'$.
			Hence, 
			\begin{align*}
				\dist(u^t_1,v) 	& = \dist(u^t_1,u^t_{i',j',1},v) \\
								& \overset{(\ref{eq:dist-utij1-to-ut1})}{=} (i'-1)y + j' + x + i + j' - y \\
								& = (j-3)y+2j'+i+x < (j-1)y - i  \\
								& \overset{(\ref{eq:dist-utij2-to-ut1})}{=} \dist(u^t_1,u^t_{i,j,2})  < \dist(u^t_1,u).
			\end{align*}
			Thus, $u$ and~$v$ are separated.

			\noindent \emph{Subcase 2:} $u \in \BR$.\\
			Again, it remains to consider~$v \in \TL'\cup \BL'$.
			Hence, $\dist(u^b_n, u) = \dist(u^b_n, r^j_i, u) \overset{(\ref{eq:dist-rij-to-ubn})}{=} (n-j+1)y + n - i + 1 + x$ where~$x = \dist(u,u^t_{i,j,2}) < y$.
			It follows that~$v \in \BL'$ and~$\dist(u^b_n,v) = \dist(u^b_n, r^{i'}_{j'},v) \overset{(\ref{eq:dist-rij-to-ubn})}{=} (n-i'+1)y + n - j' + 1 + \dist(r^{i'}_{j'},v)$.
			Assuming~$\dist(u^b_n, u) = \dist(u^b_n, v)$, we have~$j = i'+1$.
			Hence, 
			\begin{align*}
				(n-j+1)y + n - i + 1 + x 	& = (n-j+2)y + n \\&- j'  + 1 + \dist(r^{i'}_{j'},v) \\
				\dist(r^{i'}_{j'},v) 		& = x - i + j' - y.
			\end{align*}
			Since~$x < y$ it follows that~$\dist(r^{i'}_{j'},v) < j' - i$.
			Hence,
			\begin{align*}
				\dist(u^b_1,v) 	& = \dist(u^b_1,r^{i'}_{j'},v) \\
								& \overset{(\ref{eq:dist-rij-to-ub1})}{=} i'y + n - j' + 1 + x - i + j' - y \\
								& = (j-2)y + n - i + 1 + x \\& < jy + n - i + 1  \\
								& \overset{(\ref{eq:dist-rij-to-ub1})}{=} \dist(u^b_1,r^j_i) < \dist(u^b_1,u).
			\end{align*}
			Thus, $u$ and~$v$ are separated.

		\smallskip
		\noindent \textbf{Case 3 $i < i' < j' < j$:}\\
		If $u\in \TR\cup \BR$, then since $j'<j$ by Claim~2 it follows that~$\{u,v\}$ are separated.
		It thus remains the case where $u\in \M$.
		If $v\in \TL'\cup \TR'\cup \BL'\cup \BR'$, then Claim~4~\&~5 prove that $\{u,v\}$ are separated. Thus let $v\in \M'$ and $i+j=i'+j'$ (otherwise they would be separated by Claim~2).

		Assume that~$\dist(u,u^t_1) = \dist(v,u^t_1)$
			and let~$x = \dist(w^{i,j}_1,u)$. From $\dist(u^t_1,w^{i',j'}_1)-\dist(u^t_1,w^{i,j}_1)=(i'y+j')-(iy+j)=(i'-i)y+j'-j$ it follows that~$\dist(w^{i',j'}_1,v) = x - j' + j - (i'-i)y$.
			
			From \autoref{eq:dist-rij-to-ub1} follows~$\dist(u^b_1,u) = \dist(u^b_1,r^i_j)+y+x = iy+n-j+1+y+x$.
			Furthermore, it follows that
			\begin{align*}
				\dist(u^b_1,v) 	& = \dist(u^b_1,r^{i'}_{j'})+y+x-j'+j-(i'-i)y \\
				& = (i+1)y + n + j - 2j' + 1 + x. 
			\end{align*}
			Thus, $\dist(u^b_1,u) - \dist(u^b_1,v)  = (i+1)y+n-j+1+x  - ((i+1)y + n + j - 2j' + 1 + x)= 2j' - 2j$.
			Since~$j \neq j'$ it follows that~$u$ and~$v$ are separated.\oldqed
\end{proof}
}
We now have all the ingredients to prove \autoref{prop:all-but-vertex-is-separated}. 
%
{\begin{proof}[Proof of \autoref{prop:all-but-vertex-is-separated}]
  We will show that for each pair of vertices in~$G'$ except for a pair $\{l^i_j,r^i_j\}$ for some $1\le i,j\le n$ there is a vertex in~$\{u^t_\ell,u^t_r,u^b_\ell,u^b_r\}$ that separates it. We have three groups of vertices in~$G'$ namely vertex-gadget-vertices, edge-gadget vertices, and vertices in the skeletal structure (consisting of top- and bottom-line). Next, we shall show that each of them is separated from all others.
 
 \emph{Skeletal vertices:} We prove that all vertices on the top- and bottom-line vertices are separated from all others. 
 First, for each vertex pair from the skeletal structure by \autoref{lem:skel-distances} there is a shortest path between two vertices from $\{u^t_\ell,u^t_r,u^b_\ell,u^b_r\}$ that contains both vertices, implying that they are separated.
 
 Next, consider a vertex pair $\{u,v\}$ where~$u$ is contained in the skeletal structure and~$v$ is contained in a vertex-gadget~$g^V_s$. More specifically, let $u$~be on the top-line between~$u^t_i$ and~$u^t_{i+1}$ (the proof is completely analogous for the bottom-line). If $i\le s$, then by \autoref{lem:skel-distances}(\ref{prop:dist-on-one-line}) the following path is a shortest path $u^t_\ell-u^t_{i+1}-u^t_{s}-a^t_s-v-a^b_s-u^b_s-u^b_{s+1}-u^b_r$.
  Symmetrically, if $i>s$, then the following is a shortest path $u^b_\ell-u^b_s-a^b_s-v-a^t_s-u^t_s-u^t_{i}-u^t_r$. 
In both cases the vertex pair lie on a shortest path starting in $\{u^t_\ell,u^t_r,u^b_\ell,u^b_r\}$ and thus is separated.

 Finally, by \autoref{lem:skel-distances} for each vertex~$u$ in the skeletal structure there is a path from $u^t_\ell$ to~$u^b_r$ that contains~$u$ and there is no shortest path containing any edge-gadget vertex, implying that~$u$ is separated by~$u^t_\ell$ or~$u^b_r$ from all edge-gadget vertices.

\emph{Vertex-gadget:} By the argument above, vertex-gadget vertices are separated from vertices in the skeletal structure. Furthermore, for each vertex-gadget vertex~$v$ there is a shortest path from~$u^t_\ell$ to~$u^b_r$ via~$v$ and no shortest path between them contains any edge-gadget vertex, implying that $u^t_\ell$ or $u^b_r$ separate~$v$ from all edge-gadget vertices.

It remains to prove that any vertex-gadget vertex~$v$ is separated from any other vertex-gadget vertex~$v'$ except in the case that they correspond to a pair $\{l^i_j,r^i_j\}$. Consider first the subcase where $v$ and~$v'$ are contained in the same vertex-gadget, say~$g^v_i$. Then, by \autoref{lem:skel-distances}(\ref{prop:dist-on-different-lines}) the following is a shortest path:
$u^t_\ell-u^t_s-a^t_s-r^s_1-r^s_2-\ldots-r^s_n-a^b_s-u^b_s-u^b_r.
$
Clearly, the subpath $r^s_1-r^s_2-\ldots-r^s_n$ can be exchanged by $l^s_1-l^s_2-\ldots-l^s_n$. This implies that~$v$ and~$v'$ are separated.

Consider the subcase where $v$ and $v'$ are in different vertex-gadgets, say $v\in g^V_i$ and $v'\in g^V_j$ with $i<j$. By \autoref{lem:skel-distances}(\ref{prop:dist-on-different-lines}) the following paths are shortest paths:
$u^t_\ell-u^t_i-a^t_i-v$ and
$u^t_\ell-u^t_j-a^t_j-v'$.
Thus $\dist(u^t_\ell,v)\le 2+(i-1)y+y+n$ and $\dist(u^t_\ell,v')>(j-1)y+y$, implying that $\dist(u^t_\ell,v)\neq \dist(u^t_\ell,v')$.

 \emph{Edge-gadget:} Because of the above considerations it is enough to prove that edge-gadget vertices are separated from other edge-gadget vertices. This is done by  \autoref{lem:twoEdgeGadgetSeparated}. 
\oldqed
\end{proof}
}

\subsection{Correctness of the Reduction}\label{subsec:correctness}
Based on \autoref{obs:traversing-edge-gadget}, \autoref{lem:traversing-edge-gadgets}, and \autoref{prop:all-but-vertex-is-separated} we next prove the correctness of our reduction (see \autoref{sec:construction}). For the sake of readability the proof is separated into two implications.

\begin{proposition}\label{prop:hinrichtung}
	If $(G,h)$ is a yes-instance of \BipDomSet, then $(G',k)$ is a yes-instance of \MD.
\end{proposition}
{\begin{proof}
For a yes-instance $(G=(V,E),h)$  with $V=\{v_1,\ldots,v_n\}$ of \BipDomSet denote by $K\subseteq V$ a dominating set of size at most~$h$. We prove that the corresponding \MD instance~$(G'=(V',E'),k)$ with $k=h+4$ is also a yes-instance. More specifically, we prove that the set~$L\subseteq V'$ that contains $\{u^t_\ell,u^t_r,u^b_\ell,u^b_r\}$ and for each vertex $v_i\in K$ the vertex~$r^i_1$ is a metric basis.

\looseness=-1 By \autoref{prop:all-but-vertex-is-separated} the vertices $\{u^t_\ell,u^t_r,u^b_\ell,u^b_r\}\subseteq L$ separate all pairs of vertices in~$V'$ except the vertex pair $\{l^i_j,r^i_j\}$ for all $1\le i,j\le n$. 
Clearly, all pairs  $\{l^i_j,r^i_j\}$ are separated if $r^i_1\in L$. 
Thus, consider the case where $r^i_1\notin L$. As~$K$ is a dominating set there is a vertex~$r^{\alpha}_1\in L$ such that $\{v_i,v_\alpha\}\in E$, implying that there is an edge-gadget~$g^E_{i,\alpha}$.
Next, we prove that $r^{\alpha}_1$ separates the pair $\{l^i_j,r^i_j\}$ for all $1\le j\le n$. 
This is done by proving that $P^l:=r^{\alpha}_1-r^{\alpha}_i-w_2^{i,{\alpha}}-w_1^{i,{\alpha}}-r^i_\alpha - r^i_j$ (if $\alpha<i$ then interchange $w^{i,\alpha}_1$ and~$w^{i,\alpha}_2$) is a shortest path and all other paths between $r^{\alpha}_1$ and~$r^i_j$ are more than $\halfaddFactor y$ longer. 
Having proved this it follows that 
$P^l$ extended by a shortest path within~$g^V_i$ is also a shortest path for~$l^i_j$. 
Thus $\{l^i_j,r^i_j\}$ is separated by~$r^\alpha_1$.

The length of~$P^l$ is $(i-1) + (\alpha - 1) +(|{\alpha}-i|+\addFactor)y$. 
By \autoref{lem:traversing-edge-gadgets} each path that is at most $\addFactorMinusOne y$ longer than a shortest path from~$r^{\alpha}_1$ to~$r^i_j$ traverses at most one edge-gadget and each path that traverses an edge-gadget different from $g^E_{i,{\alpha}}$ is at least by~$2y$ longer than~$P^l$. 
Thus it remains to consider the paths from $r^{\alpha}_1$ to~$r^i_j$ that do not traverse any edge-gadget. 
There are only two of them, one following the top-line, $r^{\alpha}_1-a^t_{\alpha}-u^t_{\alpha}-u^t_i-a^t_i-r^i_j$,  and the other following the bottom-line, $r^{\alpha}_1-a^b_{\alpha}-u^b_{\alpha}-u^b_i-a^b_i-r^i_j$. 
Both are of length more than $(|{\alpha}-i)|+2)y$ and, thus, are at least $\halfaddFactor y$ longer than~$P^l$. 
Thus~$P^l$ is a shortest path, implying that $r^{\alpha}_1$ separates $\{l^i_j,r^i_j\}$.
\oldqed
\end{proof}}%
\begin{proposition}\label{prop:rueckrichtung}
	If $(G',k)$ is a yes-instance of \MD, then $(G,h)$ is a yes-instance of \BipDomSet.
\end{proposition}
{\begin{proof}
Let $(G'=(V',E'),k)$ by a yes-instance of~\MD where $G'$ is constructed from the \BipDomSet instance~$(G=(V,E),h)$ with $k=h+4$ and $V=\{v_1,\ldots,v_n\}$. 
Furthermore, let~$L$ be a metric basis of~$G'$ of size at most~$k$. As already argued, $L$~contains at least one degree-one neighbor of each of the endpoints $\{u^t_1,u^t_n,u^b_1,u^b_n\}$ (otherwise the degree-one neighbors would not be separated). 
Then \autoref{prop:all-but-vertex-is-separated} proves that these degree-one neighbors separate all vertices in~$G'$ except the vertex pairs $\{l^i_j,r^i_j\}$ for all $1\le i,j\le n$.

We now form a vertex subset~$K\subseteq V$ and prove that it is a dominating set of size at most~$h$: 
For each vertex $v\in L$ that is contained in a vertex-gadget~$g^V_i$ add $v_i\in V$ to~$K$. 
Additionally, for each vertex $v\in L$ contained in an edge-gadget~$g^E_{i,j}$ with $i<j$ add $v_i$ to~$L$ if $v$ is contained on the $\TL$- or $\BL$-part of~$g^E_{i,j}$ and add $v_j$ to~$L$ in all other cases.

We next prove that~$K$ is a dominating set for~$G$. 
Suppose towards a contradiction that there is a vertex~$v_i\in V$ that is not dominated by~$K$. 
By definition of~$K$ none of the vertices in~$g^V_i$ is contained in the metric basis~$L$. 
However, there is one vertex~$u\in L$ that separates $\{l^i_1,r^i_1\}$. 
Denote by $\P^l$ the set of all shortest paths from~$u$ to~$l^i_1$ and by $\P^r$ the set of all shortest paths from $u$ to $r^i_1$. 
Observe that~$l^i_1$ and~$r^i_1$ both have the same distance to~$a^t_i$ and~$a^b_i$ and that each path in~$\P^l$ either contains~$a^t_i$ or~$a^b_i$.
Thus all paths in~$\P^r$ neither contain~$a^t_i$ nor~$a^b_i$, since otherwise~$l^i_1$ and~$r^i_1$ would not have been separated by~$u$.
Hence, each path in~$\P^r$ enters~$g^V_i$ via an entrance vertex~$r^i_j$ of an edge-gadget~$g^E_{i,j}$. 
If $u$ is contained either directly in one of these edge-gadgets or it is contained in~$g^V_j$, then by the construction of~$K$ this implies that either $v_i$ or~$v_j$ is contained in~$K$. 
This yields a contradiction since $\{v_i,v_j\}\in E$ and thus $v_i$ is dominated.

Towards a contradiction, consider a shortest path in~$P\in \P^r$ entering $g^V_i$ via~$r^i_j$ but $u$~is  neither contained in $g^E_{i,j}$ nor in~$g^V_j$. 
Clearly, by \autoref{obs:traversing-edge-gadget} it follows that $P$~traverses~$g^E_{i,j}$.
Thus, $P$ enters~$g^E_{i,j}$ via~$u^t_{i,j,2}$ ($u^t_{i,j,1}$ if~$i > j$) or~$r^j_i$.
However, by \autoref{lem:skel-distances} the shortest path from~$u^t_{i,j,2}$ ($u^t_{i,j,1}$) to~$r^i_j$ contains~$a^t_i$ ($a^b_i$), implying a contradiction in the first case.
Hence, we can assume that~$P$ enters~$g^E_{i,j}$ via~$r^j_i$.
By \autoref{lem:traversing-edge-gadgets}  it traverses only~$g^E_{i,j}$, implying that it enters $g^V_j$ either via an anchor or via some $r^j_\alpha$. 
If~$P$ enters~$g^V_j$ via the anchor~$a^t_j$ ($a^b_j$) this implies that the path from~$u$ to~$r^i_1$ contains~$u^t_j$ ($u^b_j$). 
However, by \autoref{lem:skel-distances} the shortest path from~$u^t_j$ ($u^b_j$) to~$r^i_1$ contains~$a^t_i$ ($a^b_i$), yielding a contradiction.
In the remaining case the path from~$u$ to~$r^i_1$ enters~$g^V_j$ via~$r^j_\alpha$ and since it traverses only~$g^E_{i,j}$, this implies that~$u$ is contained in~$g^E_{j,\alpha}$. 
In addition, by the construction of~$K$ it follows that~$u$ has distance greater than~$y$ to~$r^j_\alpha$ and, hence, $P$ contains~$w^{j,\alpha}_1$ or~$w^{j,\alpha}_2$. 
The subpath from~$w^{j,\alpha}_1$ or from $w^{j,\alpha}_2$ to~$r^i_1$ is of length at least $(1+|j-i|+\addFactor)y$. 
However, either $w^{i,j}_1$ or $w^{i,j}_2$ has distance at most $y+n$ to~$u^t_j$ and $\dist(u^t_j,r^i_1)=(|j-i|+1)y+1$, implying that~$P$ is not a shortest path.\oldqed
\end{proof}}%

\autoref{prop:hinrichtung} together with \autoref{prop:rueckrichtung} imply that our reduction given in \autoref{sec:construction} is correct. 
Additionally, observe that the maximum degree in any graph constructed by our reduction is three. In the remaining part we discuss the computation lower bounds that are implied by it.

\section{W[2]-Completeness}\label{sec:w2-completness}

In the previous section we proved the correctness of our reduction which maps an instance $(G,h)$ of \BipDomSet into an instance $(G',k)$  of \MD with $k=h+4$. 
Since \BipDomSet is W[2]-hard with respect to~$h$~\cite{RS08Algorithmica}, this implies that \MD is W[2]-hard with respect to~$k$ on maximum degree three graphs. 
Note that this classification is tight in the sense that \MD is (trivially) polynomial-time solvable on graphs with maximum degree two. 
We prove in this section that \MD is indeed W[2]-complete.

\begin{theorem}\label{thm:w2-complete}
	\MD on graphs with bounded degree three is  W[2]-complete with respect to the parameter size of a metric basis.
\end{theorem}
{\begin{proof}
The W[2]-hardness follows from the discussion above.
Hence, it remains to show containment in W[2].
This is done by giving a parameterized reduction to the W[2]-complete \RedBlueDomSet problem~\cite{DF99}: 
Given a bipartite graph~$(R\cup B,E)$ and an integer~$h \geq 1$ it is asked whether there is a size at most~$h$ vertex subset $D\subseteq R$ that dominates all vertices in~$B$.
%
 
For an instance $(G=(V,E),k)$ of \MD we construct an equivalent \RedBlueDomSet instance $(G'=(R\cup B,E'),k)$ as follows: First the vertex set~$B$ is formed by inserting for each vertex pair $\{u,w\}\subseteq V$ a vertex $\alpha_{u,w}$. Then $R$ is a copy of~$V$ and there is an edge between $v\in R$ and $\alpha_{u,w}\in B$ if $\dist(v,u)\neq \dist(v,w)$. 
It is straightforward to argue that there is a one-to-one correspondence between the vertices in a metric basis for~$G$ and a red-blue dominating set in~$G'$.\oldqed
\end{proof}
}

\section{Running Time and Approximation Lower Bounds}\label{sec:run-lower-bound}
\looseness=-1 We next show a running time as well as an approximation lower bound for \MD.
 
\citet{CCF+05} proved that \DomSet (given an $n$-vertex graph, decide whether it has a size-$h$ dominating set) cannot be solved in $n^{o(h)}$ time, unless $\text{FPT}=\text{W[1]}$. 
By the details of the reduction in~\cite{RS08Algorithmica} (there is a one-to-one correspondence between the solution sets) this also holds for \BipDomSet. 
This implies together with the observation that the parameter~$k$ in our reduction (see \autoref{sec:construction}) is linearly upper-bounded  by the parameter~$h$ from the \BipDomSet instance where we reduce from, the same running-time lower bound for \MD.

\begin{theorem}\label{thm:eth-lower-bound}
 Unless $\text{FPT}=\text{W[1]}$, \MD cannot be solved in $n^{o(k)}$ time, even on maximum degree three graphs.
\end{theorem}
Note that the lower bound provided by \autoref{thm:eth-lower-bound} is asymptotically tight in the sense that a trivial brute-force algorithm that tests each size-$k$ vertex subset whether it is a metric basis achieves a running time of~$O(n^{k+2})$.

Additionally, observe that the proof of \autoref{prop:rueckrichtung} also provides a one-to-one correspondence between a metric basis and a dominating set in the instance where we reduce from. Moreover, our reduction can be computed in polynomial time. The reduction from \DomSet to \BipDomSet~\cite{RS08Algorithmica} also admits these two properties. Thus, the result that \DomSet cannot be approximated within~$o(\log n)$, unless $\text{NP}=\text{P}$~\cite{AS03}, transfers to \MD.

\begin{theorem}
 Unless $\text{NP}=\text{P}$, \MD on maximum degree three graphs cannot be approximated within a factor of~$o(\log n)$.
\end{theorem}

\section{Conclusion}
 We have shown that \MD is W[2]-complete even on graphs with maximum degree three. 
By modifying our construction appropriately
we conjecture that it is possible to show that \MD is W[2]-complete even on bipartite graphs with maximum degree three.

We performed a first step towards a systematic study of the parameterized complexity of \MD. From our perspective, the most interesting questions that arise is whether \MD is fixed-parameter tractable on planar graphs or with respect to the treewidth of the input graph.
By simple observations on vertices with the same neighborhood, it is straightforward to argue that \MD is fixed-parameter tractable with respect to the size of a vertex cover. 
This motivates a systematic study of ``stronger parameterizations''~\cite{KN12}, for instance the size of a feedback vertex set.
Finally, we would like to mention the  open question whether the $2^{o(n)}$ lower bound for \DomSet (unless the exponential time hypothesis fails) can be transfered to~\MD. 

\paragraph{\bf Acknowledgements.}
We thank Rolf Niedermeier for helpful comments improving the presentation.

{\footnotesize
\bibliographystyle{abbrvnat}
\bibliography{bibliography}
}
\appendix
\newpage
\section{Proofs}

\appendixProofText
\appendixProofTextT
\end{document}